\documentclass[conference]{IEEEtran}
\IEEEoverridecommandlockouts
\usepackage{cite}
\usepackage{amsmath,amssymb,amsfonts,amsthm}
\usepackage{algorithm}
\usepackage{algorithmic}
\usepackage{graphicx}
\usepackage{textcomp}
\usepackage{xcolor}
\usepackage{bm}
\DeclareMathOperator{\diag}{diag}
\DeclareMathOperator{\Tr}{Tr}
\DeclareMathOperator{\Arg}{Arg}
\DeclareMathOperator{\Expectation}{\mathop{\mathbb{E}}}
\newcommand\myciteup[1]{{\setcitestyle{square,super}\cite{b1}}}
\def\BibTeX{{\rm B\kern-.05em{\sc i\kern-.025em b}\kern-.08em
    T\kern-.1667em\lower.7ex\hbox{E}\kern-.125emX}}
\newtheorem{Beamforming Design}{Proposition}
\newtheorem{Proposition}{Proposition}

\begin{document}

\title{Secure Transmission Strategy for Intelligent Reflecting Surface Enhanced Wireless System}
\author{\IEEEauthorblockN{Biqian Feng, Yongpeng Wu and Mengfan Zheng}
\thanks{B. Feng, and Y. Wu are with the Department of Electronic Engineering, Shanghai Jiao Tong University, Minhang 200240, China (e-mail: fengbiqian@sjtu.edu.cn; yongpeng.wu@sjtu.edu.cn;).
	(Corresponding author: Yongpeng Wu.).

M. Zheng is with the Department of Electrical and Electronic Engineering, Imperial College London, London SW7 2AZ, U.K. (email: m.zheng@imperial.ac.uk).
}
}
\maketitle

\begin{abstract}
In this paper, we investigate the design of secure transmission frameworks with an intelligent reflecting surface (IRS). Our design aims to minimize the system energy consumption in cases of rank-one and full-rank access point (AP)-IRS links. To facilitate the design, the problem is divided into two parts: design of beamforming vector at AP and phase shift at IRS. In the rank-one channel model, the beamforming vector design and phase shift design are independent. A closed-form expression of beamforming vector is derived. Meanwhile, some algorithms, including the semidefinite relaxation algorithm and projected gradient algorithm, are taken to solve the phase shift problem in the case of instantaneous channel, and in the statistical channel model, the impact of phase shift on the overall system is analyzed. However, since beamforming and phase shift depend on each other in the full-rank model, we refer to conventional wiretap model and utilize an eigenvalue-based algorithm to obtain beamforming vector, while the aforementioned two phase optimization schemes are also applied.  Simulation results show that the IRS-enhanced system is envisioned to improve physical layer security.

\end{abstract}

\begin{IEEEkeywords}
intelligent reflecting surface, physical layer security, beamforming vector, phase shift
\end{IEEEkeywords}

\section{Introduction}
As a new technology to deal with weak link problem in scattering environment, intelligent reflecting surface (IRS) has received growing attention in recent years. Reference \cite{Q.Wu1} first proposed that IRS was envisioned to improve spectrum and energy efficiency in wireless communication. Then, reference \cite{Q.Wu2} introduced discrete phase shift optimization in IRS for practical systems. Furthermore, in contrast to previous works based on instantaneous channel state information (CSI), reference \cite{LIS_Statiscal} evaluated the performance of IRS under statistical CSI assumption. Additionally, reference \cite{main_reference} studied MIMO transmission schemes in the presence of a Line-of-Sight (LoS) link between access point (AP) and IRS.

Physical layer secure transmission techniques have been investigated for a long time. If the eavesdropper happens to have a better channel than the receiver, artificial-noise aided approaches can be used to confuse the eavesdropper \cite{Propose_AN}. Reference \cite{AN_power} further studied optimization of power consumption with semidefinite programming (SDP)-based algorithm. Cooperative jamming is another solution to enhance secrecy performance at the legitimate users \cite{closed-form}. As far as we have investigated, existing works on secure transmission techniques for IRS mainly focused on the instantaneous CSI problem with majorization-minimization Algorithm \cite{Secure_IRS1}\cite{Secure_IRS2}. In order to supplement the transmission properties of statistical channels, this paper focuses on the rank-one channel transmission performance.

The main contributions of this paper are as follows. (1) In the case of rank-one channel model in AP-IRS link, we derive the optimal solution of beamforming vector and phase shift designs and emphatically analyse the impact of phase shift on the overall system. (2) In other cases, we derive the optimal beamforming vector and propose two algorithms to handle the phase problem. The transmission process is decomposed into two stages, including beamforming vector design and phase shift design. More specifically, we investigate rank-one and full-rank channels between AP and IRS. On one hand, due to the good property of the rank-one channel, we separate beamforming vector and phase shift to achieve low complexity. Whether it is an instantaneous channel or a statistical channel in IRS-User link or IRS-Eve link, the beamforming matrix has the same form, which only relates to the AP-IRS link. Furthermore, we analyse the impact of phase shift on the whole system. On the other hand, in full-rank channel, we refer to conventional wiretap model to adopt eigenvalue-based algorithm and SDP/(projected gradient descend) PGD algorithm to optimize it. Simulation results show that the proposed algorithms yield great performance. Due to the blocked AP-User/Eve link, eavesdropper close to AP cannot wiretap efficiently while user close to IRS can achieve great quality of service. Our results show that IRS-enhanced system can improve physical layer security well.

The rest of this paper is organized as follows. In Section \uppercase\expandafter{\romannumeral2}, the system model and problem formulation are introduced. In Section \uppercase\expandafter{\romannumeral3}, we analyse the advantage of rank-one channel and show a reasonable scheme to design beamforming vector and phase shift. Section \uppercase\expandafter{\romannumeral4} provides another scheme for full-rank channels. Simulation results and discussions are given in Section \uppercase\expandafter{\romannumeral5}. Section \uppercase\expandafter{\romannumeral6} concludes the paper.

The notation of this paper are as follows. Boldface lowercase and uppercase letters, such as $\boldsymbol{a}$ and $\boldsymbol{A}$, are used to represent vectors and matrices, respectively. $\boldsymbol{I}_n$ denotes the n-by-n identity matrix. Superscripts $T$ and $H$ stand for the transpose and the conjugate transpose, respectively. $\vert *\vert$ denotes the Euclidean norm of a complex vector. $\lambda_{max}\left(\boldsymbol{A}\right)$ and $\gamma_{max}\left(\boldsymbol{A}\right)$ respectively denote the maximum eigenvalue of matrix $\boldsymbol{A}$ and its corresponding eigenvector. $\Arg\left(\boldsymbol{v}\right)$ denotes the phases of complex elements in the vector $\boldsymbol{v}$. $\mathcal{CN}(\boldsymbol{\mu},\boldsymbol{\Sigma})$ denotes a complex circular Gaussian distribution with mean $\boldsymbol{\mu}$ and covariance $\boldsymbol{\Sigma}$. 

\section{System Model And Problem Formulation}
As shown in Fig. \ref{Gaussian MISO wiretap channel with IRS}, we consider a Gaussian multiple input single output (MISO) wiretap channel model based on intelligent reflecting surface enhanced (IRS-enhanced) link. In this model, AP is equipped with $M$ antennas, while the legitimate user and the eavesdropper both have only a single antenna. Compared with the traditional wiretap channel model, the IRS-enhanced system introduces an IRS device, which is an intelligent control system that can dynamically adjust the phase through passive beamforming according to changes in the environment to upgrade communication quality.
\begin{figure}[htbp]
	\includegraphics[scale=0.4]{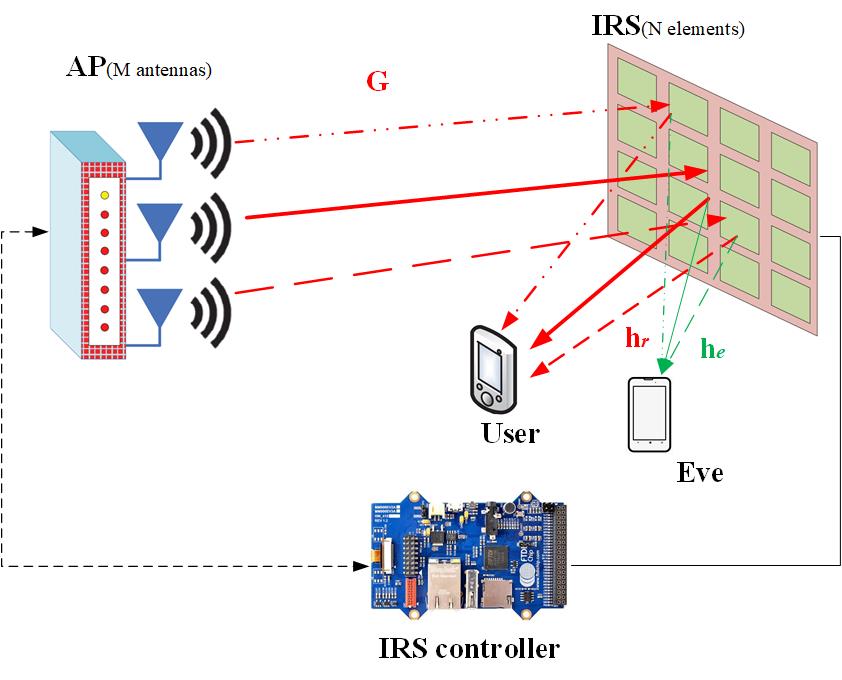}
	\caption{Gaussian MISO wiretap channel with IRS}
	\label{Gaussian MISO wiretap channel with IRS}
\end{figure}
\subsection{System Model}
Conventional MISO wiretap channel model only has an AP-User link and an AP-Eve link \cite{MIMO_I}. Once the AP-User link is blocked by an obstacle, the quality of communications will rapidly decline. As micro electromechanical systems develop rapidly, IRS is introduced to alleviate the impact of weak AP-User channel. Here, we assume quasi-static flat-fading channels, and the received complex baseband signals at the legitimate receiver and the eavesdropper are respectively given by\cite{Q.Wu1}:
\begin{equation}
\label{received_signals}
\begin{aligned}
&y_r = \sqrt{\alpha_{r}}\sqrt{P}\boldsymbol h_{r}^{H}\boldsymbol\Theta \boldsymbol G^H\boldsymbol\omega s +n_r,\\
&y_e = \sqrt{\alpha_{e}}\sqrt{P}\boldsymbol h_{e}^{H}\boldsymbol\Theta \boldsymbol G^H\boldsymbol\omega s + n_e,
\end{aligned}
\end{equation}
where $\alpha_r, \alpha_e$ are the channel attenuation coefficient for the legitimate user and the eavesdropper, respectively. $P$ is the signal power assigned to the target user. $\boldsymbol h_r\sim\mathcal{CN}(\boldsymbol{0},\sigma_{h_r}^{2}\boldsymbol{I_N}), \boldsymbol h_e\sim\mathcal{CN}(\boldsymbol{0},\sigma_{h_e}^{2}\boldsymbol{I_N})$ are the reflection channels of the desired user and the eavesdropper, respectively. $\boldsymbol\Theta = \diag\{e^{j\theta_1}, e^{j\theta_2},...,e^{j\theta_N}\}$, where $\theta_n$ is the phase shift introduced by the $n$th elements of IRS. $\boldsymbol G$ is the channel matrix between AP and IRS. $\boldsymbol\omega$ is the beamforming vector with $\vert\boldsymbol\omega\vert^2=1$. $s$ is the origin signal satisfying $\mathbb{E}\{\left|s\right|^2\}=1$. $n_r$, $n_e$ are additive white Gaussian noise with variance one at the legitimate receiver and eavesdropper, respectively.

The secure channel capacity in \cite{MIMO_I} can be generalized to the IRS-enhanced channel model. The channel capacity under three different CSI assumptions is shown as follows.
\begin{itemize}
\item Full CSI of both legitimate and eavesdropper channels is known to AP and IRS:
\begin{equation}
\label{C_Perfect_CSI_User_Eve}
\begin{aligned}
C=\log (1+\alpha_{r}P\vert &\boldsymbol h_{r}^{H}\boldsymbol\Theta \boldsymbol G^H\boldsymbol\omega \vert^2) \\
&-\log (1+\alpha_{e}P\vert \boldsymbol h_{e}^{H}\boldsymbol\Theta \boldsymbol G^H\boldsymbol\omega \vert^2 )
\end{aligned}
\end{equation}

\item Only full CSI of the legitimate channel is known to AP and IRS:
\begin{equation}
\label{C_Perfect_CSI_Eve}
\begin{aligned}
C=\log (1+&\alpha_{r}P\vert \boldsymbol h_{r}^{H}\boldsymbol\Theta \boldsymbol G^H\boldsymbol\omega \vert^2 )\\
&-\Expectation_{\boldsymbol{h}_e}\left(\log (1+\alpha_{e}P\vert \boldsymbol h_{e}^{H}\boldsymbol\Theta \boldsymbol G^H\boldsymbol\omega \vert^2 )\right)
\end{aligned}
\end{equation}

\item Only statistical information on the legitimate channel is known to AP and IRS:
\begin{equation}
\label{C_Statistical_CSI_User_Eve}
\begin{aligned}
C=\Expectation_{\boldsymbol{h}_r}(\log &(1+\alpha_{r}P\vert \boldsymbol h_{r}^{H}\boldsymbol\Theta \boldsymbol G^H\boldsymbol\omega \vert^2 ))\\
&-\Expectation_{\boldsymbol{h}_e}\left(\log (1+\alpha_{e}P\vert \boldsymbol h_{e}^{H}\boldsymbol\Theta \boldsymbol G^H\boldsymbol\omega \vert^2 )\right)
\end{aligned}
\end{equation}

\end{itemize}
A practical scenario of IRS-enhanced wiretap channel is when the BS attempts to transmit a private message to some users and treats other users as eavesdroppers, i.e., the eavesdropper is an idle user of the system.
\subsection{Problem Formulation}
Let $R$ denote the secure capacity requirement of the whole system. In this paper, our goal is to design an effective scheme to minimize power by adjusting beamforming vector and phase shift. The corresponding optimization problem is formulated as:
\begin{subequations}
\begin{align}
&(P1)&\underset{\boldsymbol\omega,\boldsymbol\Theta}{\min} & &P&\\
& &s.t. & & &\frac{1+\alpha_{r}P\vert \boldsymbol h_{r}^{H}\boldsymbol\Theta \boldsymbol G^H\boldsymbol\omega \vert^2 }{1+\alpha_{e}P\vert \boldsymbol h_{e}^{H}\boldsymbol\Theta \boldsymbol G^H\boldsymbol\omega \vert^2 } \geq 2^R,\label{Capacity_Constraint}\\
& & & & &\theta_n \in [0, 2\pi), n=1,2...N,\label{Phase_Constraint}\\
& & & & &\vert \boldsymbol{\omega} \vert = 1.\label{Beamforming_Constraint}
\end{align}
\end{subequations}

Problem $(P1)$ is a non-convex problem due to the non-concave objective function with respect to $\boldsymbol\omega$ and $\boldsymbol\theta$. Unfortunately, there is no standard method for solving the problem. In the sequel, we will analyze the design schemes for the cases when the channel $\boldsymbol{G}$ is either a rank-one matrix or a full-rank matrix. We will present some conclusions with given instantaneous CSI or statistical CSI in the following section.

\section{Joint Design for Rank-One $\boldsymbol{G}$}
IRS is expected to be deployed on high building near BS with no obstacle between AP and IRS, so the corresponding channel matrix $\boldsymbol{G}$ is of rank-one \cite{rank_one}. For convenience, let us denote
\begin{equation}
\label{Rank_One_Channel}
\boldsymbol{G}=\boldsymbol{a}\boldsymbol{b}^H,
\end{equation}
where $\boldsymbol{a}$ and $\boldsymbol{b}$ are deterministic vectors. Then, the channel matrix $\boldsymbol{G}$ in problem $(P1)$ is replaced by $\boldsymbol{a}\boldsymbol{b}^H$ and the power optimization problem $(P1)$ can be written as follows:
\begin{subequations}
	\begin{align}
	&(P2)&\underset{\boldsymbol\omega,\boldsymbol\Theta}{\min} & &P&\\
	& &s.t. & & &\frac{1+\alpha_{r}P\vert \boldsymbol h_{r}^{H}\boldsymbol\Theta  \boldsymbol{b} \boldsymbol{a}^H\boldsymbol\omega \vert^2 }{1+\alpha_{e}P\vert \boldsymbol h_{e}^{H}\boldsymbol\Theta \boldsymbol{b} \boldsymbol{a}^H\boldsymbol\omega \vert^2 } \geq 2^R,\label{Capacity_Constraint_P2}\\
	& & & & &\theta_n \in [0, 2\pi), n=1,2...N,\label{Phase_Constraint_P2}\\
	& & & & &\vert \boldsymbol{\omega} \vert = 1.\label{Beamforming_Constraint_P2}
	\end{align}
\end{subequations}
The computationally efficient approaches to solve $(P2)$ will be provided later.
\subsection{Beamforming Design}
In fact, when $\boldsymbol{G}$ is a rank-one matrix, we find that the optimal beamforming vector is exclusively related to the AP-IRS channel.
\begin{Proposition}
	The setting of a rank-one channel between the AP and the IRS can simplify the solution by decoupling the beamforming vector and phase shift design. The optimal beamforming vector is $\boldsymbol{\omega}^* = \frac{\boldsymbol{a}}{\vert \boldsymbol{a} \vert}$.
\end{Proposition}
\begin{proof}
	Starting with the constraint, $(\ref{Capacity_Constraint_P2})$ can be written as
	\begin{equation}
	\label{p_closed_form}
	P \geq \frac{2^R-1}{\vert\boldsymbol{a}^H\boldsymbol\omega \vert^2 \left(\alpha_{r}\vert\boldsymbol h_{r}^{H}\boldsymbol\Theta \boldsymbol b \vert^2-2^R\alpha_{e}\vert\boldsymbol h_{e}^{H}\boldsymbol\Theta \boldsymbol b \vert^2\right)}.
	\end{equation}
	Clearly, the design of beamforming in $\vert\boldsymbol{a}^H\boldsymbol\omega \vert$ and phase shift in $\alpha_{r}\vert\boldsymbol h_{r}^{H}\boldsymbol\Theta \boldsymbol b \vert^2-2^R\alpha_{e}\vert\boldsymbol h_{e}^{H}\boldsymbol\Theta \boldsymbol b \vert^2$ are independent in formula (\ref{p_closed_form}). In other words, there exists a best solution to achieve the maximum of the above two terms simultaneously. Thus, we will only take $\vert\boldsymbol{a}^H\boldsymbol\omega \vert$ into account when designing the beamforming vector. Therefore, the optimal beamforming vector is given by
	\begin{equation}
	\label{w_closed_form1}
	\boldsymbol{\omega}^* = \frac{\boldsymbol{a}}{\vert \boldsymbol{a} \vert}.
	\end{equation}
	This completes the proof.
\end{proof}
	A closed-form solution has been proposed to design beamforming vector if all channels are fully known to both the AP and legitimate users in \cite{closed-form}. The solution can be generalized to the IRS-enhanced system and we have
	\begin{equation}
	\begin{aligned}
	\label{w_closed_form2}
	\boldsymbol\omega&=\gamma_{max} \left[\alpha_{r}\vert\boldsymbol{a}\boldsymbol{b}^H\boldsymbol{\Theta}^H\boldsymbol h_r\vert^2-2^R\alpha_{e}\vert\boldsymbol{a}\boldsymbol{b}^H\boldsymbol{\Theta}^H\boldsymbol h_e\vert^2\right].\\
	&=\gamma_{max}\left[ \left(\alpha_{r}\vert\boldsymbol{b}^H\boldsymbol{\Theta}^H\boldsymbol h_r\vert^2-2^R\alpha_{e}\vert\boldsymbol{b}^H\boldsymbol{\Theta}^H\boldsymbol h_e\vert^2\right) \boldsymbol{a}\boldsymbol{a}^H \right]\\
	&=\gamma_{max}\left[\boldsymbol{a}\boldsymbol{a}^H \right]\\
	&=\mathop{\max_{\vert\boldsymbol{\omega}\vert=1}}\boldsymbol{\omega^H}\boldsymbol{a}\boldsymbol{a^H}\boldsymbol{\omega}\\
	&=\mathop{\max_{\vert\boldsymbol{\omega}\vert=1}}\vert \boldsymbol{a}^H\boldsymbol{\omega}\vert.
	\end{aligned}
	\end{equation}

The proposed scheme of (\ref{w_closed_form1}) is consistent with eigenvalue-based algorithm \cite{closed-form} in essence in the case of the rank-one channel.
\subsection{Phase Shift Design}
In the previous subsection, we show that the rank-one channel of the AP-IRS link has the advantage to optimally design beamforming and phase shift conveniently. Hence, we only concentrate on the second term in the denominator on the right-hand side of (\ref{p_closed_form}). The optimization problem about phase shift design is given by:
\begin{subequations}
\begin{align}
&(P3)&\mathop{\max}_{\boldsymbol\Theta}& &\alpha_{r}\vert\boldsymbol h_{r}^{H}\boldsymbol\Theta \boldsymbol b \vert^2-2^R\alpha_{e}\vert\boldsymbol h_{e}^{H}\boldsymbol\Theta \boldsymbol b \vert^2&\\
&&s.t.& &\theta_n \in [0, 2\pi), n=1,2...N.\quad\quad&
\end{align}
\end{subequations}
Let $\boldsymbol{v}=\left[e^{j\theta_1},e^{j\theta_2},...,e^{j\theta_N}\right]^T, \boldsymbol{B}=\diag(\boldsymbol{b})$. Then, the problem (P3) can be written as:
\begin{subequations}
	\begin{align}
	(P4)&\mathop{\max}_{\boldsymbol v} &\boldsymbol v^H\boldsymbol B^H\left(\alpha_{r}\boldsymbol h_r\boldsymbol h_{r}^{H}-2^R\alpha_{e}\boldsymbol h_{e}\boldsymbol h_{e}^{H}\right) \boldsymbol B\boldsymbol v&\\
	&s.t. &\theta_n \in [0, 2\pi), n=1,2...N.\quad\quad\quad\quad\qquad&
	\end{align}
\end{subequations}

Define $\boldsymbol A\triangleq \boldsymbol B^H\left(\alpha_{r}\boldsymbol h_r\boldsymbol h_{r}^{H}-2^R\alpha_{e}\boldsymbol h_{e}\boldsymbol h_{e}^{H}\right)\boldsymbol B$. The problem has been discussed in \cite{NP-hard}, which proved that this is an NP-hard problem. Two common solutions are provided to solve above problem $(P4)$, including SDP and PGD.
\begin{itemize}
	\item $(SDP)$ Note that $\boldsymbol v^H(-\boldsymbol A)\boldsymbol v =
	-\Tr\left(\boldsymbol A\boldsymbol V\right)$, where $\boldsymbol V = \boldsymbol v\boldsymbol v^H$. Clearly, $\boldsymbol V$ is a positive
	semidefinite matrix, i.e., $\boldsymbol V \succeq 0$, and $rank(\boldsymbol V) = 1$. By relaxing
	the rank-one constraint on $\boldsymbol V$, we have
	\begin{equation*}
	\begin{aligned}
	(P4)\qquad\mathop{\min}_{\boldsymbol V} &\quad -\Tr(\boldsymbol{A}\boldsymbol{V})\\
	s.t.&\quad \diag(\boldsymbol{V})=1,\\
	&\quad \boldsymbol{V} \succeq \boldsymbol{0}.
	\end{aligned}
	\end{equation*}
	Problem $(P4)$ is a standard semidefinite problem (SDP) and can be effectively solved via CVX software. But the key point is to obtain the near optimal solution $\boldsymbol{v}$ from $\boldsymbol{V}^*$. Anthony \cite{QCQP_SDP} and Goemans \cite{QCQP_SDP_Geoemans} demonstrated their algorithms to obtain good approximation guarantees for the model. More details are provided in Algorithm 1.
	
	\item $(PGD)$ We employ gradient search to monotonically decrease the objective function. The derivative of objective function can be expressed as
	\begin{equation}
	\label{derivative}
	\frac{\partial \left(-\boldsymbol{v}^H\boldsymbol{A}\boldsymbol{v}\right)}{\partial \theta_i}=j\sum_{n=1}^{N}A_{in}e^{j\left(\theta_n-\theta_i\right)}-j\sum_{m=1}^{N}A_{mi}e^{j\left(\theta_i-\theta_n\right)}
	\end{equation}
	The objective function decreases fastest if one goes from $\boldsymbol{v}$  in the direction of the negative gradient $\boldsymbol{p}$, in which $\boldsymbol{p}_i=\frac{\partial \left(-\boldsymbol{v}^H\boldsymbol{A}\boldsymbol{v}\right)}{\partial \theta_i}$. It follows that,
	\begin{equation*}
	\begin{aligned}
	\boldsymbol v'_{k+1}&=\boldsymbol v'_{k} -\mu \boldsymbol p_k,\\
	\boldsymbol v_{k+1}&=e^{j \Arg\left(\boldsymbol v'_{k+1}\right)},
	\end{aligned}
	\end{equation*}
	where $\boldsymbol v'_{k}$ denotes the induced phases at step $k$ and $\boldsymbol{p}^k$ is the adopted ascent direction at step $k$. $\mu$ is a suitable step size. More details are provided in Algorithm 2.
\end{itemize}

\begin{algorithm}
	\caption{Joint Design with SDP Algorithm, Rank-One Channel $\boldsymbol G$ }
	\begin{algorithmic}[1]
		\REQUIRE All channel state information $\boldsymbol G = \boldsymbol a \boldsymbol b^H$, $\boldsymbol h_{r}$, $\boldsymbol h_{e}$
		\ENSURE Beamforming vector $\boldsymbol{\omega}$ and Phase Shift $\boldsymbol\Theta$;
		\STATE Compute beamforming vector with channel state information $\boldsymbol G$ according to $\boldsymbol{\omega}^* = \frac{\boldsymbol{a}}{\vert \boldsymbol{a} \vert}$;
		\STATE Solve the problem $(P4)$ with CVX and obtain an optimal solution $\boldsymbol V^*$;
		\STATE Since $\boldsymbol V^*$ is positive semidefinite, we can obtain a Eigendecomposition $\boldsymbol V^*=\boldsymbol U^*\boldsymbol \Sigma \boldsymbol U$ and $\boldsymbol \Sigma \succeq 0$;
		\STATE Obtain suboptimal solution $\boldsymbol v=\boldsymbol U^*\boldsymbol \Sigma^\frac{1}{2}\boldsymbol r$, $\boldsymbol r\sim\mathcal{CN}(\boldsymbol{0},\boldsymbol{I})$;
		\STATE Set $\boldsymbol{\Theta}=\diag\left(e^{j\Arg\left(\boldsymbol{v}\right)}\right)$.
	\end{algorithmic}
\end{algorithm}

\begin{algorithm}
	\caption{Joint Design with PGD Algorithm, Rank-One Channel G }
	\begin{algorithmic}[1]
		\REQUIRE All channel state information $\boldsymbol G = \boldsymbol a \boldsymbol b^H$, $\boldsymbol h_{r}$, $\boldsymbol h_{e}$
		\ENSURE Beamforming vector $\boldsymbol{\omega}$ and Phase Shift $\boldsymbol\Theta$;
		\STATE Beamforming vector design is the same as Algorithm 1;
		\STATE Set $\boldsymbol v$ as a random vector with each element $v_k \in [-\pi, \pi)$
		\REPEAT
			\STATE Compute $\boldsymbol{p}^k$ according to (\ref{derivative});
			\STATE $\mu$=backtrack line search \cite{PGD};
			\STATE $\boldsymbol v'_{k+1}=\boldsymbol v'_{k} -\mu \boldsymbol p_k$;
		\UNTIL $\vert\boldsymbol p_k\vert<\epsilon$
		\STATE Set $\boldsymbol{\Theta}=\diag\left(e^{j\Arg(\boldsymbol v'_{k+1})}\right)$.
	\end{algorithmic}
\end{algorithm}

\subsection{Statistical Channel Model}
\begin{Proposition}
	If statistical CSI of eavesdropper and full CSI of legitimate user are provided, the optimal beamforming vector and phase shift is $\boldsymbol{\omega}^* = \frac{\boldsymbol{a}}{\vert \boldsymbol{a} \vert}$ and $\boldsymbol{v} = e^{-j\Arg\left(\diag(\boldsymbol{h}_r^H)\boldsymbol{b}\right)+\boldsymbol \alpha}$, respectively. $\boldsymbol\alpha$ is a vector with the same elements.
\end{Proposition}
\begin{proof}
	Consider $\boldsymbol h_e\sim\mathcal{CN}(\boldsymbol{0},\sigma_{h_e}^{2}\boldsymbol I_N)$.
	\begin{equation}
	\begin{aligned}
	C&=\log \left(1+\alpha_{r}P\vert \boldsymbol h_{r}^{H}\boldsymbol\Theta \boldsymbol b\vert^2 \vert\boldsymbol a^H\boldsymbol\omega \vert^2 \right)\\
	&\quad\quad\quad\quad-\Expectation_{\boldsymbol{h}_e}\left(\log \left(1+\alpha_{e}P\vert \boldsymbol h_{e}^{H}\boldsymbol\Theta \boldsymbol b\vert^2\vert\boldsymbol a^H\boldsymbol\omega \vert^2 \right)\right)\\
	&=\log \left(1+\alpha_{r}P\vert \boldsymbol h_{r}^{H}\boldsymbol\Theta \boldsymbol b\vert^2 \vert\boldsymbol a^H\boldsymbol\omega \vert^2 \right)\\
	&\quad\quad\quad\quad-\Expectation_{\boldsymbol{h}_e}\left(\log \left(1+\alpha_{e}P\vert \boldsymbol h_{e}^{H}\boldsymbol b\vert^2\vert\boldsymbol a^H\boldsymbol\omega \vert^2 \right)\right)
	\end{aligned}
	\end{equation}
	For any given phase shift, we have
	\begin{equation}
	\begin{aligned}
	\frac{\partial C}{\partial \vert\boldsymbol a^H\boldsymbol\omega \vert^2}&=\frac{\alpha_{r}P\vert \boldsymbol h_{r}^{H}\boldsymbol\Theta \boldsymbol b\vert^2}{1+\alpha_{r}P\vert \boldsymbol h_{r}^{H}\boldsymbol\Theta \boldsymbol b\vert^2 \vert\boldsymbol a^H\boldsymbol\omega \vert^2}\\
	&\quad\quad-\Expectation_{\boldsymbol{h}_e}\left(\frac{\alpha_{e}P\vert \boldsymbol h_{e}^{H}\boldsymbol b\vert^2}{1+\alpha_{e}P\vert \boldsymbol h_{e}^{H}\boldsymbol b\vert^2\vert\boldsymbol a^H\boldsymbol\omega \vert^2}\right) \geq 0
	\end{aligned}
	\end{equation}
	Note that $C$ is an increasing function. The optimal beamforming vector makes $\vert\boldsymbol a^H\boldsymbol\omega \vert^2$ achieve maximum. Similar to section III-A, the optimal bramforming vector is $\boldsymbol{\omega}^* = \frac{\boldsymbol{a}}{\vert \boldsymbol{a} \vert}$. 
	
	Thus, for a fixed eavesdropper rate of $F_1(\alpha_{e}\alpha_{h_e}\sigma_{h_e}^{2}P\vert \boldsymbol a \vert^2 \vert \boldsymbol b\vert^2)$, where $F_1(x)$ is defined in Lemma 3 of \cite{statistical_channel}, $\vert \boldsymbol h_{r}^{H}\boldsymbol\Theta \boldsymbol b\vert$ should be maximized. So, the optimal phase shift is $\boldsymbol{v} = e^{-j\Arg\left(\diag(\boldsymbol{h}_r^H)\boldsymbol{b}\right)+\boldsymbol \alpha}$, where $\boldsymbol\alpha$ is a vector with the same elements.
	
	This completes the proof.
\end{proof}
\begin{Proposition}
	If only statistical CSI of both the legitimate user and the eavesdropper are provided, the optimal beamforming vector is $\boldsymbol{\omega}^* = \frac{\boldsymbol{a}}{\vert \boldsymbol{a} \vert}$ and the mathematical expectation of the maximum achievable secrecy rate is $F_1(\alpha_{r}\alpha_{h_r}\sigma_{h_r}^{2}P\vert \boldsymbol a \vert^2 \vert \boldsymbol b\vert^2)-F_1(\alpha_{e}\alpha_{h_e}\sigma_{h_e}^{2}P\vert \boldsymbol a \vert^2 \vert \boldsymbol b\vert^2)$, irrespective of phase shift $\boldsymbol{\Theta}$.
\end{Proposition}
\begin{proof}
	Due to the assumption of $\boldsymbol h_r\sim\mathcal{CN}(\boldsymbol{0},\sigma_{h_r}^{2}\boldsymbol I_N)$ and $\boldsymbol h_e\sim\mathcal{CN}(\boldsymbol{0},\sigma_{h_e}^{2}\boldsymbol I_N)$ , we have
	\begin{equation}
	\label{statistical channel model with removed theta}
	\begin{aligned}
	C&=\Expectation_{\boldsymbol{h}_r}\left(\log \left(1+\alpha_{r}P\vert \boldsymbol h_{r}^{H}\boldsymbol\Theta \boldsymbol b\vert^2 \vert\boldsymbol a^H\boldsymbol\omega \vert^2 \right)\right)\\
	&\quad\quad\quad-\Expectation_{\boldsymbol{h}_e}\left(\log \left(1+\alpha_{e}P\vert \boldsymbol h_{e}^{H}\boldsymbol\Theta \boldsymbol b\vert^2\vert\boldsymbol a^H\boldsymbol\omega \vert^2 \right)\right)\\
	&\overset{(a)}{=}\Expectation_{\boldsymbol{h}_r}\left(\log \left(1+\alpha_{r}P\vert \boldsymbol h_{r}^{H}\boldsymbol b\vert^2 \vert\boldsymbol a^H\boldsymbol\omega \vert^2 \right)\right)\\
	&\quad\quad\quad-\Expectation_{\boldsymbol{h}_e}\left(\log \left(1+\alpha_{e}P\vert \boldsymbol h_{e}^{H}\boldsymbol b\vert^2\vert\boldsymbol a^H\boldsymbol\omega \vert^2 \right)\right)
	\end{aligned}
	\end{equation}
	(a) indicates the setting of phase shift does not affect the expectation of secrecy capacity. For any phase shift, we have
	\begin{equation}
	\begin{aligned}
	\frac{\partial C}{\partial \vert\boldsymbol a^H\boldsymbol\omega \vert^2}&=\Expectation_{\boldsymbol{h}_e}\left(\frac{\alpha_{r}P\vert \boldsymbol h_{r}^{H} \boldsymbol b\vert^2}{1+\alpha_{r}P\vert \boldsymbol h_{r}^{H} \boldsymbol b\vert^2 \vert\boldsymbol a^H\boldsymbol\omega \vert^2}\right)\\
	&\quad\quad-\Expectation_{\boldsymbol{h}_e}\left(\frac{\alpha_{e}P\vert \boldsymbol h_{e}^{H}\boldsymbol b\vert^2}{1+\alpha_{e}P\vert \boldsymbol h_{e}^{H}\boldsymbol b\vert^2\vert\boldsymbol a^H\boldsymbol\omega \vert^2}\right) \geq 0
	\end{aligned}
	\end{equation}
	Similar to proposition 2, the optimal beamforming vector is $\boldsymbol{\omega}^* = \frac{\boldsymbol{a}}{\vert \boldsymbol{a} \vert}$. Substituting $\boldsymbol{\omega}$ in (\ref{statistical channel model with removed theta}) with the optimal beamforming vector, we have
	\begin{equation}
	\label{optimal statistical channel model}
	\begin{aligned}
	C = F_1(\alpha_{r}\alpha_{h_r}\sigma_{h_r}^{2}P\vert \boldsymbol a \vert^2 \vert \boldsymbol b\vert^2)-F_1(\alpha_{e}\alpha_{h_e}\sigma_{h_e}^{2}P\vert \boldsymbol a \vert^2 \vert \boldsymbol b\vert^2)
	\end{aligned}
	\end{equation}
	This completes the proof.
\end{proof}

\section{Joint Design for Full-Rank G}
The above work makes reasonable use of the advantages of rank-one channel in design. However, when we take the Rician channel model or the Rayleigh channel model, $\boldsymbol{G}$ will not be rank-one any more. In this section, we will discuss the method to design beamforming vector and phase shift with full-rank channel $\boldsymbol{G}$.

Similarly, we deduce the power condition from $(P1)$ as
\begin{equation}
\label{p_close_form2}
P \geq \frac{2^R-1}{\boldsymbol{\omega}^H\boldsymbol{G}\boldsymbol{\Theta}^H\left(\alpha_{r}\boldsymbol{h}_r\boldsymbol{h}_r^H-2^R\alpha_{e}\boldsymbol{h}_e\boldsymbol{h}_e^H\right)\boldsymbol{\Theta}\boldsymbol{G}^H\boldsymbol{\omega}}.
\end{equation}
Hence, the objective function in $(P1)$ can be converted to the form of $\boldsymbol\omega$ and $\boldsymbol\Theta$.
\begin{subequations}
	\begin{align}
	&(P5)&\mathop{\max}_{\boldsymbol\omega,\boldsymbol\Theta} &\quad\boldsymbol{\omega}^H\boldsymbol{G}\boldsymbol{\Theta}^H\left(\alpha_{r}\boldsymbol{h}_r\boldsymbol{h}_r^H-2^R\alpha_{e}\boldsymbol{h}_e\boldsymbol{h}_e^H\right)\boldsymbol{\Theta}\boldsymbol{G}^H\boldsymbol{\omega}&\\
	&&s.t. &\quad\quad\theta_n \in [0, 2\pi), n=1,2...N.&\\
	&&&\quad\quad\vert\boldsymbol\omega\vert=1.&
	\end{align}
\end{subequations}

Compared with (\ref{p_closed_form}), $\boldsymbol{G}$ and $\boldsymbol{\Theta}$ interact with each other and cannot be separated in (\ref{p_close_form2}). Next, we will introduce some schemes to optimize it.
\subsection{Beamforming Design}
Let
\begin{subequations}
	\begin{align}
\boldsymbol h_r^{'H} &= \boldsymbol{h}_r^H\boldsymbol{\Theta}\boldsymbol{G}^H\\
\boldsymbol h_e^{'H} &= \boldsymbol{h}_e^H\boldsymbol{\Theta}\boldsymbol{G}^H
	\end{align}
\end{subequations}
denote the IRS-enhanced system channels for the AP-IRS-User link and the AP-IRS-Eve link, respectively.

For any given $\boldsymbol\Theta$, channels of AP-IRS-User link and the AP-IRS-Eve link are fixed. Obviously, it reduces to a standard system with a single legitimate user and a single eavesdropper. Reference \cite{closed-form} further derived the closed-form solution of the problem based on the dual problem and KKT conditions.
We apply the eigenvalue-based algorithm to obtain the closed-form solution of optimal beamforming vector.

\begin{Proposition}
	For any given $\boldsymbol\Theta$, the optimal solution of $(P5)$ is given by
\begin{subequations}
	\begin{align}
		\boldsymbol\omega&=\gamma_{max} \left(\alpha_r\boldsymbol h_r^{'} \boldsymbol h_r^{'H}-2^R\alpha_e\boldsymbol h_e^{'H}\boldsymbol h_e^{'H}\right),\\
		\boldsymbol\lambda^*&=\lambda_{max} \left(\alpha_r\boldsymbol h_r^{'} \boldsymbol h_r^{'H}-2^R\alpha_e\boldsymbol h_e^{'H}\boldsymbol h_e^{'H}\right),\\
		\boldsymbol\omega^*&=\sqrt{\frac{2^R-1}{\lambda^*}}\frac{\boldsymbol \omega}{\vert \boldsymbol\omega \vert},
	\end{align}
\end{subequations}
	where $\omega^*$ represents the optimal beamforming vector with the determined phase shift.
\end{Proposition}
The proof of this proposition is similar to that of Lemma 1 in \cite{closed-form} with $\boldsymbol{h}_r^H\boldsymbol{\Theta}\boldsymbol{G}^H$ substituted by $\boldsymbol h_{r}^{'H}$.
\subsection{Phase Shift Design}
In this subsection, we optimize phase shift on the premise of a given beamforming vector. Let $\boldsymbol B' = \diag(\boldsymbol{G^H}\boldsymbol\omega)$ and $\boldsymbol A'\triangleq \boldsymbol B^{'H}\left(\alpha_{r}\boldsymbol h_r\boldsymbol h_{r}^{H}-2^R\alpha_{e}\boldsymbol h_{e}\boldsymbol h_{e}^{H}\right)\boldsymbol B'$. The problem is similar to the phase shift design in Section III-B. Therefore, SDP and PGD are also suitable for full-rank channel design.

Algorithm 3 combines the eigenvalue-based algorithm for beamforming with the SDP/PGD algorithm in the case of full-rank channel.

\begin{algorithm} 
	\caption{Joint Beamforming Vector Design in AP and Phase Shift Design in IRS, Full-Rank Channel $\boldsymbol G$}
	\begin{algorithmic}[1]
		\REQUIRE All channel state information $\boldsymbol h_{r}$, $\boldsymbol h_e$, $\boldsymbol G$
		\ENSURE Beamforming vector $\boldsymbol{\omega}$ and Phase Shift $\boldsymbol\Theta$;
		\STATE Set the initial $\boldsymbol\Theta$ as identity matrix $\boldsymbol{I}_{N}$;
		\STATE Compute beamforming vector with given $\boldsymbol\Theta$ according to the closed-form solution in Proposition 3;
		\REPEAT
		\STATE Compute phase shift similar to algorithm 1 step2-5 or algorithm 2 step 2-8
		\UNTIL {Power of beamforming vector does not change any more.}
	\end{algorithmic}
\end{algorithm}

\section{Simulation Results}
The formulation of IRS-enhanced system based on the fully cartesian coordinates has an important advantage in describing positions of all components. A uniform linear array at AP and a uniform linear array of passive reflecting elements at IRS are located at $(0, 0, 25)m$ and $(0, 100, 40)m$ respectively. In practical systems, as AP and IRS are deployed in advance, we assume the AP-IRS channel is dominated by the LoS link in rank-one channel. When $\boldsymbol{G}$ only has LoS link, we model it as the $\boldsymbol{G}=\boldsymbol{a}\boldsymbol{b}^H$. The components of $\boldsymbol{a}$ and $\boldsymbol{b}$ are written as
\begin{equation*}
\begin{aligned}
a_m&=exp\left(j2\pi \frac{d_t}{\lambda}(m-1)\sin\phi_t\sin\theta_t\right),\\
b_n&=exp\left(j2\pi \frac{d_{I}}{\lambda}(n-1)\sin\phi_{I}\sin\theta_{I}\right).
\end{aligned}
\end{equation*}
respectively, where $d_t, d_I$ are both inter-antenna separation at the AP or IRS, $\phi_t, \theta_t$ represent the LoS azimuth and elevation AoDs at the AP, and $\phi_I, \theta_I$ represent the LoS azimuth and elevation AoDs at the IRS. Here, we set $\frac{d_t}{\lambda}=\frac{d_I}{\lambda}=0.5$, and $\theta_t=\tan^{-1}\left(\frac{y_I-y_t}{40-25}\right), \theta_I=\pi-\theta_t$.

As opposed to the infinite Rician factor in the rank-one channel, the full-rank channel contains NLOS components. When $\boldsymbol{G}$ has LoS link and NLOS link simultaneously, the AP-IRS channel $\boldsymbol{G}$ is given by $\boldsymbol{G}=\sqrt{\frac{K}{1+K}}\boldsymbol{a}\boldsymbol{b}^H+\sqrt{\frac{1}{1+K}}\boldsymbol G^{NLOS}$,
where $\boldsymbol G^{NLOS}$ represents Rayleigh fading components.

For the purpose of simulating the IRS-enhanced channel, we let the eavesdropper hold still while the legitimate receiver keeps moving from AP to IRS. The simulation setup is shown in Fig.~\ref{Simulation_setup}, and some parameters are given in Table I.
\begin{figure}[htbp]
	\includegraphics[scale=0.4]{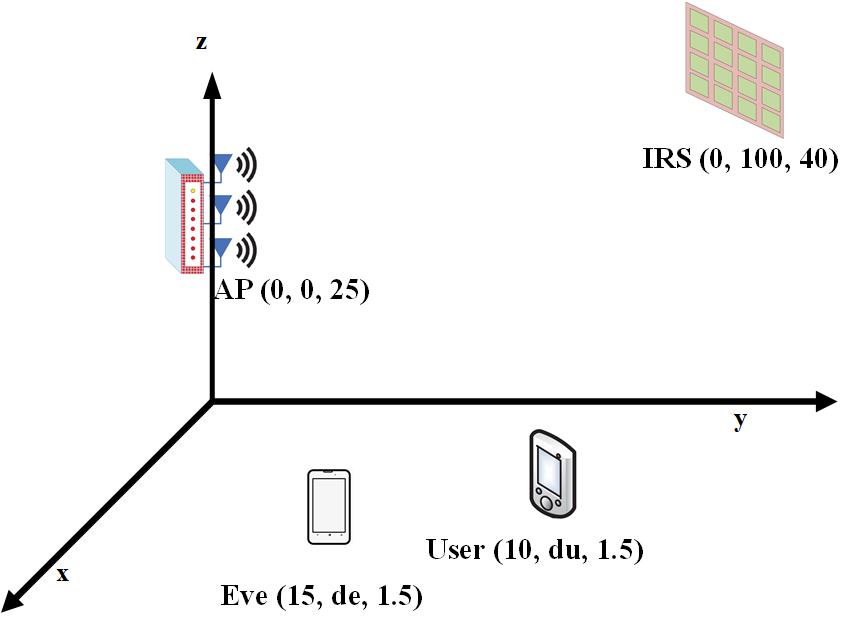}
	\caption{Simulation setup}
	\label{Simulation_setup}
\end{figure}
  
\begin{table}
\centering
\caption{Some simulation parameters}
\begin{tabular}{ |c|c| }
	\hline  
	Parameters & Values \\
	\cline{1-2}
	Number of antennas in AP, $N_t$ & 8\\
	\cline{1-2}
	Location of AP & (0, 0, 25)m \\
	\cline{1-2}
	Number of passive elements in IRS, $N_{IRS}$ & 8 \\
	\cline{1-2}
	Location of IRS & (0, 100, 40)m \\
	\cline{1-2}
	Noise variance, $\sigma^2$ & $10^{-11}$W \\
	\cline{1-2}
	Stopping criterion, $\epsilon$ & $10^{-4}$ \\
	\cline{1-2}
	Rician factor, $K$ & 2\\
	\cline{1-2}
	Elevation AoD at the AP & $\pi/2$\\
	\cline{1-2}
	Elevation AoD at the IRS & $3\pi/2$\\
	
	\hline
\end{tabular}
\end{table}

\subsection{SDP Algorithm VS PGD Algorithm}
The two algorithms are investigated in Fig. \ref{rate2power}, where the expected ergodic secrecy rate ranges from 8 $ bits/s/Hz$ to 15 $bits/s/Hz$ and the legitimate user stays 80 to 100 meters away from AP. We can see that the SDP algorithm and the PGD algorithm yield the same performance. However, the PGD algorithm has a very lower convergence speed, which indicates that PGD is not as good as SDP.
\begin{figure}[htbp]
	\centerline{\includegraphics[scale=0.4]{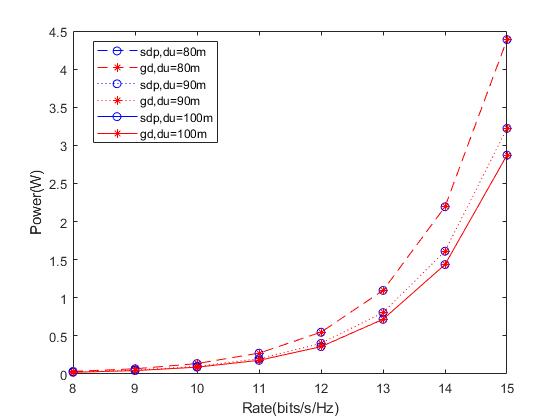}}
	\caption{Performance comparison for two algorithms}
	\label{rate2power}
\end{figure}

\subsection{Power vs Distance between AP and Legitimate User}
The second result is shown in Fig.~\ref{distance2power}, where the expected ergodic secrecy rate ranges from 12 $bits/s/Hz$ to 15 $bits/s/Hz$ and the legitimate user stays 50 to 150 meters away from AP. Compared with the first result, Fig.~\ref{distance2power} considers the distance as an independent variable in the objective function. First, it can be observed that the user far away from the IRS suffers more SNR loss due to signal attenuation and needs more power to satisfy the secrecy rate. Second, the curves are completely symmetric about the location of IRS due to the isotropic channel setting.
\begin{figure}[htbp]
	\centerline{\includegraphics[scale=0.4]{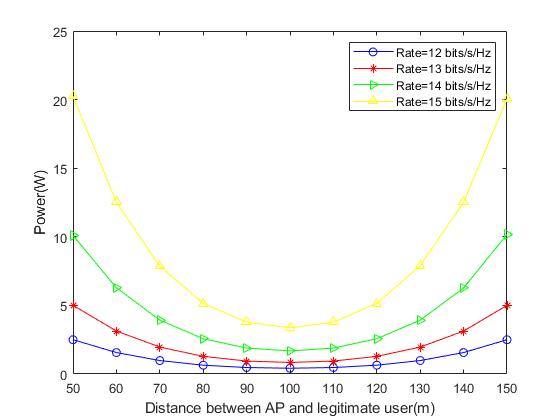}}
	\caption{Power consumption in different location}
	\label{distance2power}
\end{figure}

\section{Conclusion}
In this paper, an IRS-enhanced Gaussian MISO wiretap channel model has been studied, aiming to minimize the power while improving secrecy rate. We investigate two scenarios, including rank-one and full-rank channels. On one hand, in the rank-one channel, we separate beamforming vector and phase shift to facilitate a low-complexity design. On the other hand, in full-rank channel, we refer to conventional wiretap model to adopt eigenvalue-based algorithm and SDP/PGD algorithm to optimize it. Numerical results indicate that the proposed algorithm achieves an obvious improvement for the secrecy performance.

\end{document}